\documentclass[11pt]{article}

\usepackage{float}
\usepackage{graphicx,color}
\usepackage{epsfig}
\usepackage{amsmath,amssymb,amsthm}
\usepackage{subfigure}
\usepackage{dsfont}
\usepackage{psfrag}
\usepackage{algpseudocode}

\def\1{\mathbf{1}}

\graphicspath{{figures/}}

\newtheorem{theorem}{Theorem}[section]

\floatstyle{ruled}
\newfloat{Algorithm}{tbp}{lop}[section]

\title{Polynomial chaos based uncertainty quantification in Hamiltonian, multi-time scale, and chaotic systems}

\author{Jos{\'e} Miguel Pasini$^\dagger$ \and Tuhin Sahai$^\dagger$}

\date{$^\dagger$United Technologies Research Center \\ 411 Silver Lane, East Hartford, CT 06108, USA \\ \{pasinijm, sahait\}@utrc.utc.com}

\begin{document}

\maketitle

\begin{abstract}
Polynomial chaos is a powerful technique for propagating uncertainty through ordinary and partial differential equations. Random variables are expanded in terms of orthogonal polynomials and differential equations are derived for the expansion coefficients. Here we study the structure and dynamics of these differential equations when the original system has Hamiltonian structure, has multiple time scales, or displays chaotic dynamics. In particular, we prove that the differential equations for the expansion coefficients in generalized polynomial chaos expansions of Hamiltonian systems retain the Hamiltonian structure relative to the ensemble average Hamiltonian. We connect this with the volume-preserving property of Hamiltonian flows to show that, for an oscillator with uncertain frequency, a finite expansion must fail at long times, regardless of the order of the expansion. Also, using a two-time scale forced nonlinear oscillator, we show that a polynomial chaos expansion of the time-averaged equations captures uncertainty in the slow evolution of the Poincar{\'e} section of the system and that, as the time scale separation increases, the computational advantage of this procedure increases. Finally, using the forced Duffing oscillator as an example, we demonstrate that when the original dynamical system displays chaotic dynamics, the resulting dynamical system from polynomial chaos also displays chaotic dynamics, limiting its applicability.
\end{abstract}

\section{Introduction}

Uncertainty quantification techniques allow one to quantify output variability in the presence of parametric uncertainty. Typically, the moments of the output distributions are computed using sampling methods such as Monte Carlo~\cite{McQMc}, Quasi-Monte Carlo~\cite{QMC}, and importance sampling~\cite{Importance_Sampling}. Non-sampling approaches include response surface~\cite{Response1:book,Response2:book} and polynomial chaos based methods~\cite{Wiener}. Depending on the problem, different methods are applicable/appropriate in different scenarios. Polynomial chaos based techniques for propagating uncertainty have been used on a multitude of applications such as aeroelastic modeling~\cite{Allen2009}, transport in heterogeneous media~\cite{Ghanem1998}, Ising models~\cite{TuhinPoly}, switching systems~\cite{hybrid_uq}, combustion~\cite{Najm2009}, fluid flow~\cite{Xiu2003}, and materials models~\cite{Cit:Materials}, to name a few.

Here we study the properties and utility of using polynomial chaos expansions to propagate uncertainty through systems that have either Hamiltonian structure, multiple scales, or display chaos. We point out that polynomial chaos~\cite{Wiener} and chaos theory~\cite{Cit:Stro} are unrelated areas. Originally proposed by Nobert Wiener~\cite{Wiener} in 1938 (prior to the development of chaos theory---hence the unfortunate usage of the term \emph{chaos}), polynomial chaos expansions are a popular method for propagating uncertainty through low dimensional systems with smooth dynamics. They rely on expanding random variables in terms of orthogonal basis functions~\cite{BeyondWienerAskey}. Note that the orthogonal polynomials are chosen such that they are orthogonal to one another with respect to the prior distribution on the uncertain parameters~\cite{BeyondWienerAskey}. For example, if the underlying distribution on the uncertain parameters is Gaussian, then the associated orthogonal polynomials are Hermite polynomials~\cite{Orthopoly:book}. Similarly, if the underlying prior distribution is uniform, the associated orthogonal polynomials are Legendre~\cite{Orthopoly:book}. In general, one can construct orthogonal polynomials for arbitrary distributions~\cite{BeyondWienerAskey}. The advantage of polynomial chaos based techniques is that they provide exponential convergence for smooth processes with finite variance~\cite{cameron1947orthogonal}.

Chaos, on the other hand, refers to ``aperiodic long-term behavior in \textit{deterministic} systems that exhibits sensitive dependence on initial conditions''~\cite{Cit:Stro}. Chaos theory has been applied to a wide variety of applications such as fluid turbulence~\cite{Cit:turb}, celestial dynamics~\cite{Cit:solar}, and weather modeling~\cite{Cit:Lorenz}. It is important to point out that although the dynamics has sensitive dependence on initial conditions, it is inherently deterministic. In other words, no associated parametric uncertainty is required to observe chaos.

In this work, we present three new results. In the first part, we show that the dynamical systems that one gets by applying polynomial chaos expansions to Hamiltonian systems with uncertain parameters are also Hamiltonian. To do this, we first perform a polynomial chaos expansion of the generalized coordinates and conjugate momenta and find the evolution equations for the coefficients. We consider the expansion coefficients of the generalized coordinates as a new, larger set of ``uncertain'' generalized coordinates. By considering the averaged Hamiltonian (over parameter space) as a function of the expansion coefficients, we show that, for each of these new generalized coordinates, the coefficient in the expansion of the corresponding conjugate momentum and to the corresponding order is itself the conjugate momentum relative to the average Hamiltonian, thus demonstrating that the Hamiltonian structure in the derived differential equations is preserved. We then connect this result with the volume-preserving property of Hamiltonian flows to show that any finite polynomial chaos expansion of a harmonic oscillator with uncertain frequency must fail at long times, regardless of how many terms are kept. In the second part,  we demonstrate the application of polynomial chaos to systems with multiple time scales using perturbation theory~\cite{Cit:Rand-notes}. We demonstrate how the uncertain parameters influence the averaged dynamics of the dynamical system. In particular, we show that uncertainty can be propagated through the averaged equations instead of through the original equations, thus avoiding the computational burden of simulating a stiff system. In the third part, we apply polynomial chaos to a chaotic dynamical system (forced Duffing oscillator)~\cite{Gucken:book} and demonstrate that the resulting equations for the coefficients are also chaotic. We then show that chaotic dynamics significantly reduce the efficacy of polynomial chaos expansions at propagating uncertainty.

\section{Introduction to polynomial chaos}

Starting with a complete probability space $\Gamma$ given by
$(\Omega, \mathcal{F},\mathbb{P})$, where $\Omega$ is the sample
space, $\mathcal{F}$ is the $\sigma$-algebra on $\Omega$ and
$\mathbb{P}$ is a probability measure, let $L_{2}(\Gamma,X)$
denote the Hilbert space of square-integrable,
$\mathcal{F}$-measurable, $X$-valued random elements. Then one can,
in general, define a polynomial chaos basis
$\{\psi_k(\lambda(\omega))\}$, where $\lambda(\omega)$ is a
random vector, $\omega \in \Omega$, and $k = (k_1,k_2,\dots)$ is a
vector of non-negative indices. We denote the probability
density function of the random vector $\lambda$ by
$\rho(\lambda)$.

Generalized polynomial chaos (gPC)~\cite{BeyondWienerAskey} provides a framework for
representing second-order stochastic processes $\kappa\in
L_{2}(\Gamma,X)$ for arbitrary distributions of $\lambda$ by the following expansion:
\begin{equation}
\kappa(\lambda) =
\displaystyle\sum_{|k|=0}^{\infty}a_k \psi_k (\lambda),
\label{eq:expan1}
\end{equation}
where $|k| = \sum_{i} k_{i}$ is the sum of the indices of $k$ and
$\psi_k (\lambda)$ are orthonormal polynomials on $\Gamma$
with respect to $\rho(\lambda)$. Restricting our formalism to $\mathbb{R}^{n}$ (relevant for this work) we get the orthonormality is given by
\begin{equation}
\displaystyle\int_{\mathbb{R}^n}
\rho(\lambda)\psi_i(\lambda)\psi_k(\lambda)d\lambda =
\delta_{ik}, \label{eq:ortho}
\end{equation}
where $\delta_{ik}$ is the Kronecker delta product.
Depending on $\rho(\lambda)$ one can generate an appropriate
orthogonal basis for representing~$\kappa(\lambda)$. As mentioned earlier, if $\rho$ is
Gaussian, then the appropriate polynomial chaos basis is the set
of Hermite polynomials; if $\rho$ is the uniform distribution,
then the basis is the set of Legendre polynomials. For details
on the correspondence between distributions and polynomials
see~\cite{PolyReview,Ogura}. A framework to generate polynomials
for arbitrary distributions has been developed
in~\cite{BeyondWienerAskey}. The advantage of using polynomial chaos is that it provides exponential convergence in smooth processes~\cite{cameron1947orthogonal}. However, the approach suffers from the curse of dimensionality, making them infeasible for problems with more than a handful of parameters. To mitigate the curse of dimensionality, sparse grid techniques~\cite{Webster2007, Nobile2008, Zabaras2008}, iterative methods~\cite{surana_uq, Sahai2012, Klus2011}, regression based algorithms~\cite{blatman2010adaptive,blatman2011adaptive}, hierarchical methods~\cite{ma2009adaptive}, dimensionality reduction based techniques~\cite{ma2011kernel,marzouk2009dimensionality} have been developed.

In practice, the expansion in Eqn.~\ref{eq:expan1} is truncated
at a particular order, say,~$r$. One can then use Galerkin projections to
obtain a set of differential equations for the coefficients
$a_k$ in Eqn.~\ref{eq:expan1}~\cite{BeyondWienerAskey}. Typically, low order truncations are found to capture the uncertainty in smooth systems~\cite{cameron1947orthogonal}.

\section{Polynomial chaos based uncertainty quantification in Hamiltonian systems}

Consider a system described by the Hamiltonian $H(q,p;\lambda)$, where $\lambda$ is a vector of uncertain parameters with probability density $\rho(\lambda)$. The generalized coordinates and momenta $q_i$ and $p_i$ ($i=1,\ldots,N$) satisfy Hamilton's equations.
\begin{align*}
\dot{q}_i &= \frac{\partial H}{\partial p_i}, \\
\dot{p}_i &= -\frac{\partial H}{\partial q_i}.
\end{align*}

The generalized polynomial chaos (gPC) expansion of the coordinates and momenta is,
\begin{align*}
q_i(t;\lambda) &= \sum_k Q_{ik}(t) \psi_k(\lambda), \\
p_i(t;\lambda) &= \sum_k P_{ik}(t) \psi_k(\lambda),
\end{align*}
where the $\psi_k$ form an orthonormal basis with respect to the density $\rho$ (see Eq.~\ref{eq:ortho}).

The gPC coefficients $Q_{ik}$ and $P_{ik}$ follow deterministic equations, obtained by projecting the equations of motion along $\psi_s$
\begin{align*}
\int \dot{q}_i \psi_s \rho d\lambda &= \int \frac{\partial H}{\partial p_i} \psi_s \rho d\lambda, \\
\int \dot{p}_i \psi_s \rho d\lambda &= -\int \frac{\partial H}{\partial q_i} \psi_s \rho d\lambda.
\end{align*}
Inserting the gPC expansions and using the orthonormality condition~\eqref{eq:ortho} we obtain,
\begin{align*}
\dot{Q}_{is} &= \int \frac{\partial H}{\partial p_i} \psi_s \rho d\lambda, \\
\dot{P}_{is} &= -\int \frac{\partial H}{\partial q_i} \psi_s \rho d\lambda.
\end{align*}

Let us define the average Hamiltonian $\hat{H}$
\begin{equation*}
\hat{H} = \int H \rho d\lambda.
\end{equation*}
By using the gPC expansion of $q$ and~$p$  we can consider $\hat{H}$ as a function of $Q$ and~$P$, where $Q$ and $P$ denote the sets of coefficients $Q_{ik}$ and $P_{ik}$, respectively.

\begin{theorem}
\label{thm:hamiltonian_structure}
The gPC expansion coefficients $\{Q,P\}$ together with $\hat{H}(Q,P)$ form a Hamiltonian system, with the corresponding expansion coefficients $P_{ik}$ as conjugate momenta to~$Q_{ik}$. In other words,
\begin{align}
\dot{Q}_{ik} &= \frac{\partial \hat{H}}{\partial P_{ik}}, \label{eq:Qdot} \\
\dot{P}_{ik} &= -\frac{\partial \hat{H}}{\partial Q_{ik}}. \label{eq:Pdot}
\end{align}

\end{theorem}

\begin{proof}

We start with the right-hand side of Eq.~\eqref{eq:Qdot}:
\begin{align*}
\frac{\partial \hat{H}}{\partial P_{ik}} &= \int \sum_s \left(
\frac{\partial H}{\partial q_s} \frac{\partial q_s}{\partial P_{ik}}
+ \frac{\partial H}{\partial p_s} \frac{\partial p_s}{\partial P_{ik}}
\right) \rho d\lambda, \\
&= \int \sum_s \sum_r \frac{\partial H}{\partial p_s} \delta_{is} \delta_{kr} \psi_r \rho d\lambda, \\
&= \int \frac{\partial H}{\partial p_i} \psi_k \rho d\lambda, \\
&= \dot{Q}_{ik}.
\end{align*}
Similarly for Eq.~\eqref{eq:Pdot}:
\begin{align*}
\frac{\partial \hat{H}}{\partial Q_{ik}} &= \int \sum_s \left(
\frac{\partial H}{\partial q_s} \frac{\partial q_s}{\partial Q_{ik}}
+ \frac{\partial H}{\partial p_s} \frac{\partial p_s}{\partial Q_{ik}}
\right) \rho d\lambda, \\
&= \int \sum_s \sum_r \frac{\partial H}{\partial q_s} \delta_{is} \delta_{kr} \psi_r \rho d\lambda, \\
&= \int \frac{\partial H}{\partial q_i} \psi_k \rho d\lambda, \\
&= -\dot{P}_{ik}.
\end{align*}

\end{proof}

Note that the proof depends only on the \emph{form} of the expansion and does not require that the expansion be complete. In other words, the coefficients of a \emph{truncated} expansion will also form a (finite) Hamiltonian system relative to the average Hamiltonian when expressed as a function of the truncated expansion coefficients. Hence, polynomial chaos expansions when applied to Hamiltonian systems are also Hamiltonian. This result is not only interesting but also has practical implications. In particular, if the underlying system is Hamiltonian and one desires to propagate uncertainty using polynomial chaos, symplectic integrators~\cite{Cit:geom} will be needed to maintain numerical accuracy for long times.

We now illustrate the preservation of Hamiltonian structure on the Duffing oscillator with parametric uncertainty.

%
%

\subsection{Example: Duffing oscillator}

To provide an example of a Hamiltonian system with uncertainty, we consider the Duffing oscillator,
\begin{align}
\ddot{q} + \lambda q + q^3 = 0,
\label{eq:duffing_nf}
\end{align}
where $\lambda$ is a normally distributed uncertain parameter with mean $\mu(\lambda) = \lambda_0 = -1.0$ and standard deviation $\sigma(\lambda) = 0.1$. This system has the following Hamiltonian:
\begin{align}
H = \frac{1}{2} p^2 + \frac{\lambda}{2}q^2 + \frac{1}{4} q^4,
\label{eq:hamilton}
\end{align}
where $p=\dot q$.

The phase portrait of the undamped Duffing oscillator (in  Eq.~\ref{eq:duffing_nf}) is shown in Fig.~\ref{Fig:hamiltphase}. One can observe the Hamiltonian structure evident in phase space. In particular, the system has two centers at located at $(-1,0)$ and $(1,0)$. The equilibrium at $(0,0)$ is a saddle point. For a detailed discussion on the characteristics of the Duffing oscillator and its volume preserving flow we point the reader to~\cite{Gucken:book}.

The resulting dynamical system is of the form,
\begin{equation}
\begin{pmatrix}
\dot q\\
\dot p
\end{pmatrix}= \begin{pmatrix}
p\\
- \lambda q - q^3
\end{pmatrix}.
\label{eq:duffingsimple}
\end{equation}

Assuming that $\lambda$ is an uncertain parameter, we now perform a polynomial chaos expansion~\cite{BeyondWienerAskey} given by,
\begin{equation}
\begin{split}
q(t;\lambda) &=
\displaystyle\sum_{i=0}^{r}Q_{i}(t)\psi_{i}(\lambda),\\
p(t;\lambda) &=
\displaystyle\sum_{i=0}^{r}P_{i}(t)\psi_{i}(\lambda).
\end{split}
\label{eq:expx}
\end{equation}
By substituting the above expansion, for $r=1$, into Eq.~\ref{eq:duffingsimple} and imposing orthogonality constraints we get the following set of equations,
\begin{equation}
\begin{pmatrix}
\dot Q_{0}\\
\dot P_{0}\\
\dot Q_{1}\\
\dot P_{1}
\end{pmatrix}=\begin{pmatrix}
P_{0}\\
- \lambda_{0}Q_{0} - \sigma Q_{1} - (Q_{0}^3 + 3Q_{0}Q_{1}^2)\\
P_{1}\\
- \lambda_{0}Q_{1} - \sigma Q_{0} - 3(Q_{1}^3 + Q_{0}^2Q_{1})
\end{pmatrix}.
\label{eq:coeffs}
\end{equation}
It is easy to check that the Hamiltonian for the above system of equations is given by
\begin{equation}
\begin{split}
H_{pc} &= \frac{1}{2}P_{0}^{2} + \frac{1}{2}P_{1}^{2} + \frac{\lambda_0}{2}(Q_{0}^{2} + Q_{1}^{2}) + \sigma Q_{0}Q_{1} \\
&\quad + \frac{3}{2}Q_{0}^{2}Q_{1}^2 + \frac{1}{4} Q_{0}^4 + \frac{3}{4}Q_{1}^4.
\end{split}
\label{eq:Hpc}
\end{equation}
Similar Hamiltonians can be constructed for higher order expansions (arbitrary $r$) in the Duffing oscillator as well as for other Hamiltonian systems, such as the double pendulum and $N$~bodies interacting through Newton's law of gravitation.

\begin{figure}
  \centering
  \includegraphics[scale=0.25]{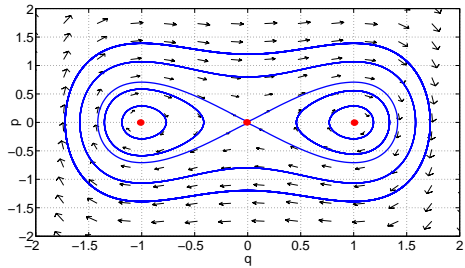}
  \caption{Phase portrait of the Duffing oscillator.\label{Fig:hamiltphase}}
\end{figure}

\subsection{Harmonic oscillator with uncertain frequency}

We have proved that the PC equations have Hamiltonian structure. We now combine this with other results in Hamiltonian theory to show that, for certain problems, the uncertainty cannot be captured by a finite PC expansion for long times, \emph{regardless of the order of the expansion}. Specifically, we focus on a harmonic oscillator with uncertain frequency and certain initial conditions:
\begin{equation}
\ddot{q} + \omega^2 q = 0 \qquad q(0) = 1 \qquad \dot{q}(0) = 0.
\label{eq:harmonic_oscillator}
\end{equation}
We choose $\omega$ uniformly distributed in $(\omega_1, \omega_2)$ and define $\omega \equiv \omega_0 + \alpha \lambda$, with $\lambda \sim U(-1,1)$, so $\alpha = (\omega_2 - \omega_1)/2$ is a measure of the magnitude of the uncertainty in frequency.

\begin{theorem}
The coefficients of a finite PC expansion of the true solution for this system converge to zero at long times.
\end{theorem}

\begin{proof}
 The solution for this system is
\begin{equation}
q(t;\lambda) = \cos (\omega_0 + \alpha \lambda) t
\label{eq:harmonic_solution}
\end{equation}
and the PC expansion in this case is
\begin{align}
q(t;\lambda) &= \sum_{k=0}^r Q_k(t) \psi_k(\lambda) \\
p(t;\lambda) &= \sum_{k=0}^r P_k(t) \psi_k(\lambda),
\end{align}
where $\psi_k(\lambda) = \sqrt{2k+1} \,{\cal P}_k(\lambda)$, and ${\cal P}_k$ is the usual Legendre polynomial of order~$k$. $\{\psi_k\}$ is a set of orthonormal polynomials in $[-1,1]$ with respect to the density $\rho(\lambda) = 1/2$. Explicitly,
\begin{align}
\psi_k(\lambda) &= \sum_{\ell=0}^k B_{k\ell} \lambda^\ell \\
B_{k\ell} &= \sqrt{2k+1} \, 2^k \binom{k}{\ell} \binom{(k+\ell-1)/2}{k}.
\end{align}
We project the PC expansion onto this basis:
\begin{align}
\int_{-1}^1 q(t;\lambda) \rho(\lambda) d\lambda &= \int_{-1}^1 \sum_{k=0}^r Q_k(t) \psi_k(\lambda) \rho(\lambda) d\lambda \\
\int_{-1}^1 p(t;\lambda) \rho(\lambda) d\lambda &= \int_{-1}^1 \sum_{k=0}^r P_k(t) \psi_k(\lambda) \rho(\lambda) d\lambda.
\end{align}
If we had an infinite-order expansion ($r=\infty$), the sum would be a series, and to exchange the integral and the series we would require uniform convergence of the sum. For any \emph{finite} sum, however, we can do the exchange and obtain these equations for the coefficients:
\begin{align}
Q_k(t) &= \int_{-1}^1 \cos[(\omega_0 + \alpha \lambda) t] \psi_k(\lambda) \rho(\lambda) d\lambda \\
P_k(t) &= \int_{-1}^1 (-\alpha \lambda) \sin[(\omega_0 + \alpha \lambda) t] \psi_k(\lambda) \rho(\lambda) d\lambda.
\end{align}
We now show that these coefficients go to zero as $t\to\infty$. Indeed,
\[
Q_k(t) = \frac{1}{2} \sum_{\ell=0}^k B_{k\ell} I_\ell,
\]
where we can integrate by parts twice to obtain a recurrence formula:
\begin{align}
I_\ell &= \int_{-1}^1 \lambda^\ell cos(\omega_0 + \alpha \lambda)t \, d\lambda \\
&= \frac{1}{\alpha t} [\sin\omega_2 t - (-1)^\ell \sin \omega_1 t] \\
&+ \frac{\ell}{(\alpha t)^2} [\cos \omega_2 t - (-1)^{\ell-1} \cos \omega_1 t ]
- \frac{\ell(\ell-1)}{(\alpha t)^2} I_{\ell-2}.
\end{align}
Since both $I_0$ and $I_1$ go to zero as $t\to\infty$, all $Q_k(t) \to 0$ as $t\to\infty$. Similarly, we can prove that all $P_k(t) \to 0$ as $t\to\infty$.

\end{proof}

With all the coefficients in any finite expansion of the true solution converging to zero as $t\to\infty$, the volume of this flow decreases. On the other hand, the solution of the PC equations obtained must preserve volume, so those coefficients cannot go to zero to match the behavior of the true solution without violating Liouville's theorem~\cite{Liboff1998}. In other words, Liouville's theorem prevents any finite PC expansion from representing the true solution of this system at long times.

\section{Polynomial chaos based uncertainty quantification in systems with multiple time scales}

Systems with multiple time scales are prevalent in a wide variety of applications related to smart grids~\cite{He2011, Parpas2011, Susuki2011}, building systems~\cite{surana_uq}, and micromechanical oscillators~\cite{Tuhin_mems1,Tuhin_mems2,Tuhin_mems3}, to name a few. Simulating these systems is challenging due to their inherent stiffness~\cite{Gear1971}. The method of multiple scales (or averaging) is a very popular approach for simulating such systems. The approach typically involves perturbing off a dynamical system whose solution can be computed in closed form~\cite{Cit:Rand-notes}. Note that this approach is applicable only in the scenario that the perturbation is small $O(\epsilon)$. The method of multiple scales captures the dynamics of the system on an $n-1$ dimensional section transversal to the flow, also known as the Poincar\'e section~\cite{Gucken:book}. For a detailed discussion on the method of multiple scales or averaging theory, see~\cite{Cit:Rand-notes,Gucken:book}.

To the best of our knowledge, no attempt has been made to extend polynomial chaos based methods to systems with multiple time scales using the method of multiple scales. Here we apply polynomial chaos to the two-time scale system given below,
\begin{align}
\ddot q + q +\epsilon\delta\dot q + \epsilon\beta q^3 = \epsilon\gamma \cos \omega t,
\label{eq:twotime_eq}
\end{align}
where $\epsilon\delta$ is the system damping, $1$ and $\epsilon\beta$ are the linear and nonlinear stiffnesses respectively, and $\epsilon\gamma$ and $\omega$ are the forcing amplitude and frequency respectively. Note that in the above system, we assume that $\epsilon$ is a small parameter (i.e. $\epsilon\ll 1$). We assume that $\gamma = \gamma_{0} + \sigma(\gamma)\eta$ is an uncertain parameter , where $\gamma_{0} = 1.0$ is the mean of $\gamma$ and $\sigma(\gamma) = 0.1$ is its standard deviation ($\eta$ is a normal random variable with zero mean and unit variance).

Using the two time scales as $\xi=\omega t$ and $\chi = \epsilon t$, one can derive the averaged equations for the system~\cite{Cit:Rand-notes,Gucken:book}. This is done by substituting $\frac{d}{dt} = \omega\frac{\partial}{\partial\xi} + \epsilon\frac{\partial}{\partial\chi}$, $\frac{d^{2}}{dt^{2}} = \omega^{2}\frac{\partial^{2}}{\partial\xi^{2}} + 2\omega\epsilon\frac{\partial^{2}}{\partial\xi\partial\chi} + \epsilon^{2}\frac{\partial^{2}}{\partial\chi^{2}}$, and $q(\xi,\chi) = q_{0}(\xi,\chi) + \epsilon q_{1}(\xi,\chi) + \hdots$ in Eqn.~\ref{eq:twotime_eq}. Collecting terms, we obtain
\begin{align}
O(1): \frac{\partial^{2}q_{0}}{\partial\xi^{2}} + q_{0} &=0,\label{eq:o1}\\
O(\epsilon): \frac{\partial^{2}q_{1}}{\partial\xi^{2}} + q_{1} &= -2\frac{\partial^{2}q_{0}}{\partial\xi\partial\chi} - \delta\frac{\partial q_{0}}{\partial\xi} - \beta q_{0}^{3} + \gamma\cos\xi \label{eq:o2}.
\end{align}

The solution to Eqn.~\ref{eq:o1}, is $q_{0}(\xi,\chi) = A(\chi)\cos\xi + B(\chi)\sin\xi$. Substituting the solution into Eqn.~\ref{eq:o2} and imposing that there be no secular terms~\cite{Cit:Rand-notes,Gucken:book} yields the averaged equations
\begin{equation}
\begin{split}
2\frac{\partial A}{\partial\chi} + \delta A -\frac{3}{4}\beta B(A^{2} + B^{2}) &= 0, \\
2\frac{\partial B}{\partial\chi} + \delta B +\frac{3}{4}\beta A(A^{2} + B^{2}) &=\gamma.
\end{split}
\label{eq:twotime_averaged}
\end{equation}
Note that the above dynamical system captures the dynamics on the Poincar\'e section of the original system~\cite{Gucken:book}. From here on we take $\beta = 1$. We will also focus on the deterministic initial condition $q=2$, $\dot{q} = 0$, so $A(0) = 2$ and $B(0) = 0$.

We now apply a polynomial chaos expansion to Eq.~\ref{eq:twotime_averaged}, to first order:
\begin{equation}
\begin{split}
A(\chi, \eta) &= a_0(\chi) H_0(\eta) + a_1(\chi) H_1(\eta) \\
B(\chi, \eta) &= b_0(\chi) H_0(\eta) + b_1(\chi) H_1(\eta),
\end{split}
\label{eq:PC_expansion_of_slow_equations}
\end{equation}
where $H_0(\eta) = 1$ and $H_1(\eta) = \eta$ are the first two probabilist's Hermite polynomials.  Substituting Eq.~\ref{eq:PC_expansion_of_slow_equations} into Eq.~\ref{eq:twotime_averaged} gives
\begin{equation}
\begin{split}
2a_0' &= - \delta a_0 + \frac{3}{4} b_0 (a_0^2 + b_0^2) \\
2b_0' &= - \delta b_0 - \frac{3}{4} a_0 (a_0^2 + b_0^2) + \gamma_0 \\
2a_1' &= - \delta a_1 + \frac{3}{4} (2a_0 b_0 a_1 + a_0^2 b_1 + 3b_0^2 b_1) \\
2b_1' &= - \delta b_1 - \frac{3}{4} (3a_0^2 a_1 + b_0^2 a_1 + 2a_0 b_0 b_1) + \sigma,
\end{split}
\label{eq:twotime_averaged_PC}
\end{equation}
with initial condition $(a_0,b_0,a_1,b_1) = (2,0,0,0)$.

For comparison purposes, we also do an equivalent polynomial chaos expansion of the original two-time equations, defining $x = q$, $y = \dot{q}$, and doing a first order expansion
\begin{align*}
x(t, \eta) &= x_0(t) H_0(\eta) + x_1(t) H_1(\eta) \\
y(t, \eta) &= y_0(t) H_0(\eta) + y_1(t) H_1(\eta).
\end{align*}
The resulting system is
\begin{equation}
\begin{split}
\dot{x}_0 &= y_0 \\
\dot{y}_0 &= -x_0 - \epsilon \delta y_0 - \epsilon x_0^3 + \epsilon \gamma_0 \cos \omega t \\
\dot{x}_1 &= y_1 \\
\dot{y}_1 &= -x_1 - \epsilon \delta y_1 - 3\epsilon x_0^2 x_1 + \epsilon \sigma \cos \omega t,
\end{split}
\label{eq:twotime_PC}
\end{equation}
with initial condition $(x_0,y_0,x_1,y_1) = (2,0,0,0)$. In what follows, we choose $\omega=1$ and, more importantly, $\delta = 0$. This is done in order to avoid having a system in which the dissipation artificially reduces the overall error.

Figure~\ref{Fig:TwoTime_comparison} shows the error in mean and standard deviation of both PC expansions compared with Monte Carlo simulations with $10^3$ samples of the original two-time system, evaluated at the Poincar\'e sections where the forcing is maximal: $t = 2\pi n$ ($n = 0, 1, 2, \ldots$). Solutions of both PC expansions and the Monte Carlo trajectories of the original system were obtained using Matlab's ode45 solver with relative tolerance of $10^{-6}$. The error has two sources: the time averaging and the truncation of the polynomial chaos expansion. As we decrease~$\epsilon$, the error from time averaging decreases and the main source of error becomes the truncation of the polynomial expansion.

As $\epsilon$ decreases, the two expansions yield increasingly similar results, but the PC expansion of the original equation becomes more expensive to compute, scaling as $1/\epsilon$, because the solver needs to trace each fast oscillation, even if we're only interested in the slow evolution of the Poincar\'e section. Figure~\ref{Fig:TwoTime_function_evaluations} shows the number of function evaluations required by this expansion as $\epsilon$ decreases (solid line) compared with the $\epsilon$-independent behavior of the averaged PC.

\begin{figure}
  \centering
  \includegraphics[width=0.5\textwidth]{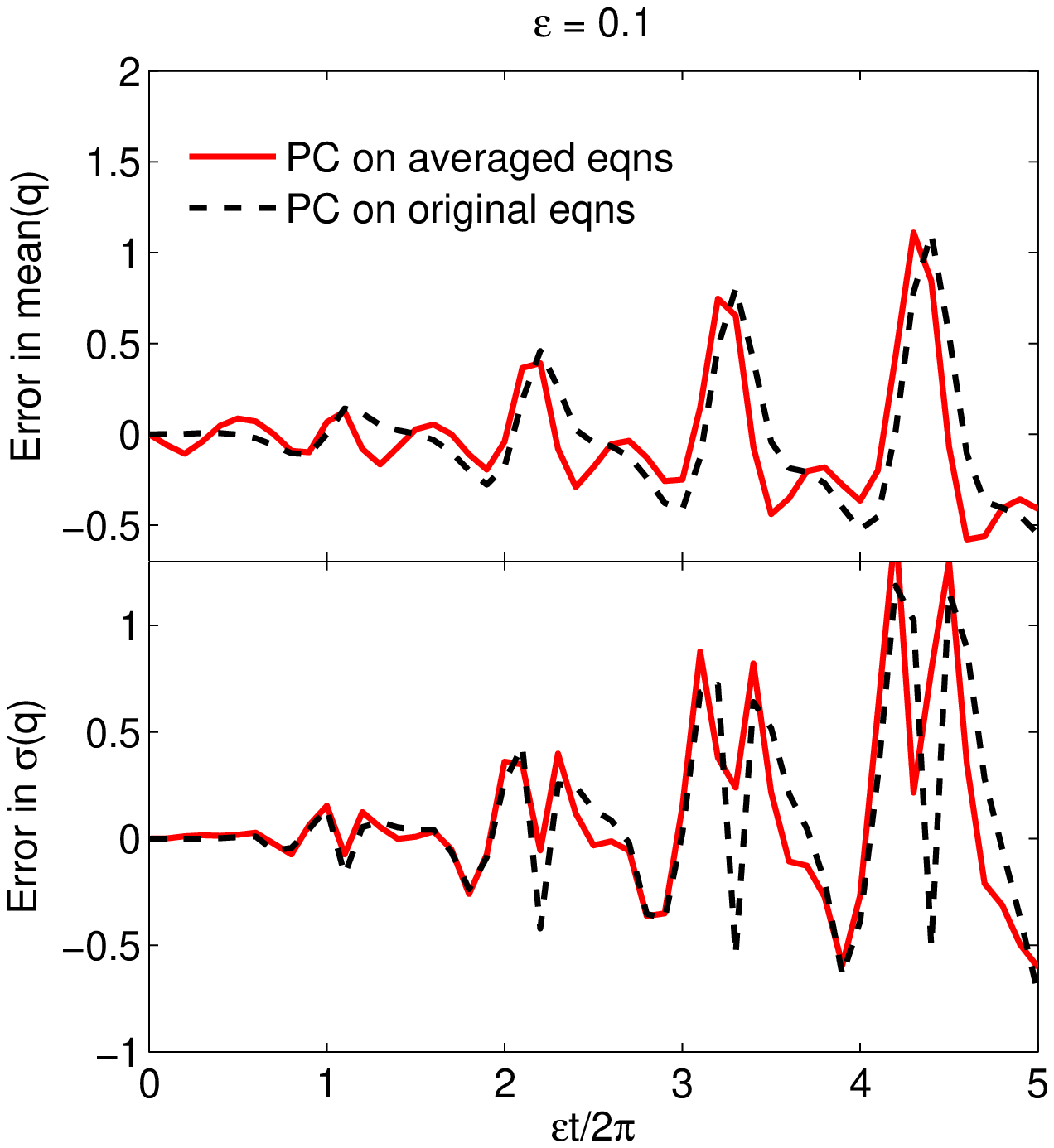}%
\includegraphics[width=0.5\textwidth]{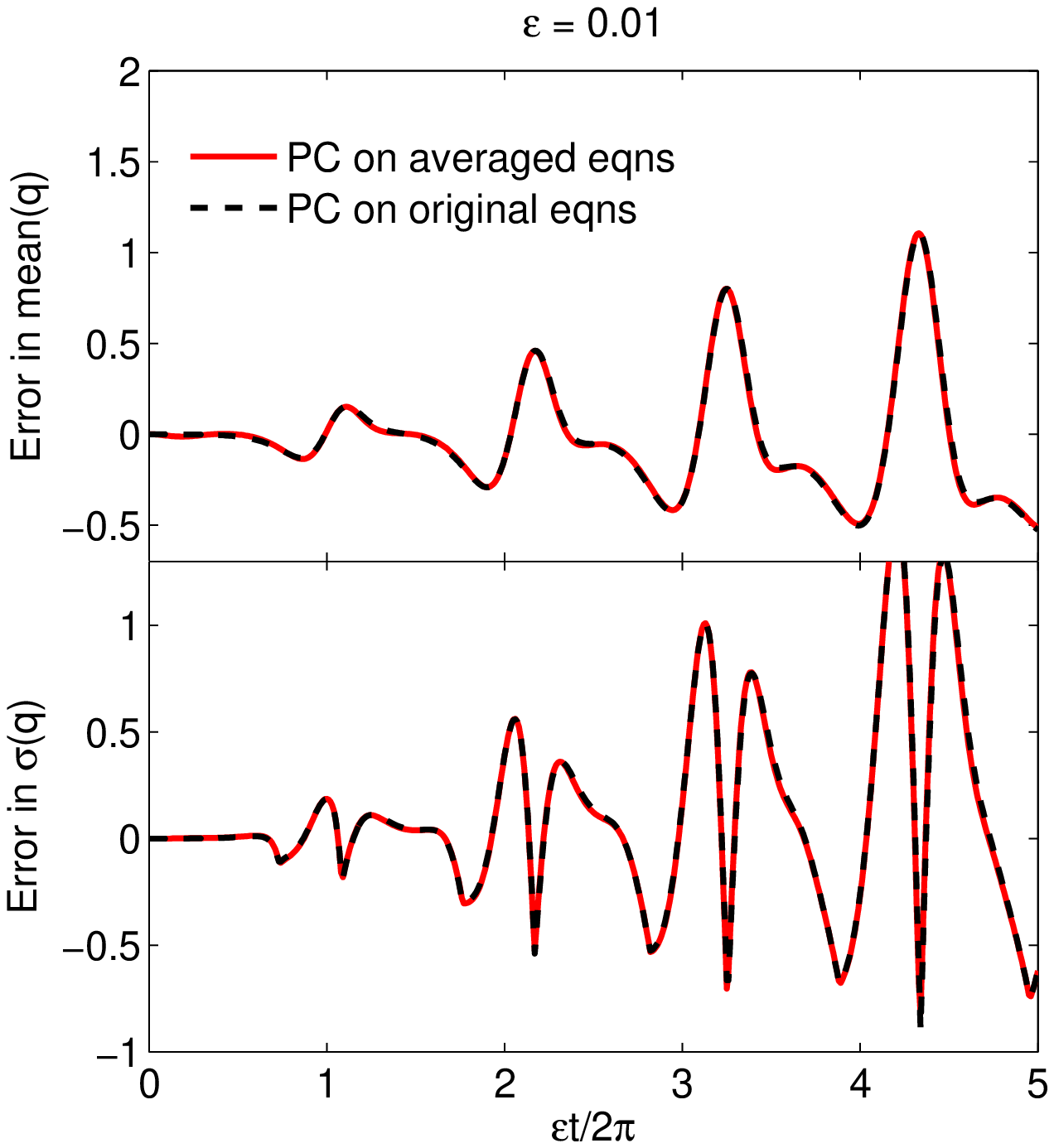}
  \caption{Absolute deviation of the polynomial chaos expansions on the averaged equations (solid line) and on the original equations (dashed line) of the two-time oscillator. The reference is a Monte Carlo simulation with 1000 samples.\label{Fig:TwoTime_comparison}}
\end{figure}

\begin{figure}
  \centering
  \includegraphics[width=0.5\textwidth]{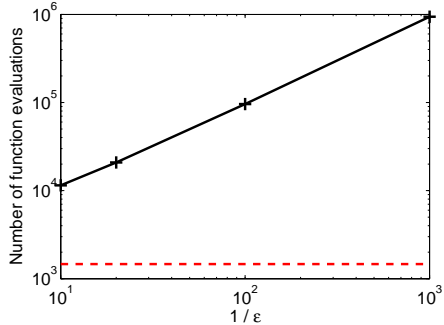}%
  \caption{Function evaluations required to solve the polynomial chaos expansion of the original two-time system as a function of the time scale separation parameter~$\epsilon$. The solution was obtained using Matlab's ode45 with relative tolerance $10^{-6}$. The dashed line shows the number of function evaluations required by the polynomial expansion on the averaged equations, which is independent of~$\epsilon$. \label{Fig:TwoTime_function_evaluations}}
\end{figure}

\section{Polynomial chaos based uncertainty quantification in chaotic systems}

We now demonstrate that dynamics of the coefficients of the polynomial chaos expansions can be chaotic, if the underlying dynamical system is chaotic. We will then show that if the underlying system is chaotic, the applicability of polynomial chaos is significantly reduced. While this is a natural result, it is not obvious. PC aims to capture the moments of the output distribution and not individual trajectories of the underlying dynamical system. Since mixing introduces averaging that could, in principle, smooth out the moments and allow them to be captured accurately, it is not obvious that chaos would necessarily make predictions worse.

For this demonstration, we pick the forced Duffing oscillator~\cite{Gucken:book} given by
\begin{equation}
\ddot q + \delta \dot q + \lambda q + q^3 = \gamma \cos \omega t,
\label{eq:duffing}
\end{equation}
where $\delta = 0.2$, $\gamma = 0.3$, $\omega = 1.0$, and $\lambda = -1.0$. Note that the
above equation (Eq.~\ref{eq:duffing}) is the same as Eq.~\ref{eq:duffing_nf} with the addition of damping and forcing terms. We can write Eq.~\ref{eq:duffing} in the form
\begin{equation}
\begin{pmatrix}
\dot q\\
\dot p
\end{pmatrix}= \begin{pmatrix}
p\\
- \delta p - \lambda q - q^3 + \gamma \cos\omega t
\end{pmatrix}.
\label{eq:duffing2d}
\end{equation}
The dynamics of the above system have been studied extensively (see, e.g., \cite{Gucken:book}). For the forced Duffing oscillator, the Poincar\'e section is given by taking ``snapshots'' of the system at phase $\phi = 0$, where $\phi = (\omega t\, \mbox{mod}\, 2\pi)$. The intersection of a single trajectory with the Poincar\'e section can be seen in Fig.~\ref{Fig:ForcedDuffing}, starting from the initial condition $(q,p)=(1,0)$.

The dynamics of the forced Duffing oscillator (at the parameter values given above) is well known to be chaotic~\cite{Gucken:book}. In fact, one can numerically compute Lyapunov exponents ($\Xi$)~\cite{Cit:Stro} for the above system and show that they are positive. Note that $\Xi>0$ is considered to be the signature of a chaotic system since it implies that the system response is sensitive to initial conditions. We find that the nominal system gives $\Xi \approx 0.93$, hence (numerically) implying the existence of chaos.

\begin{figure}
  \centering
  \includegraphics[scale=0.4]{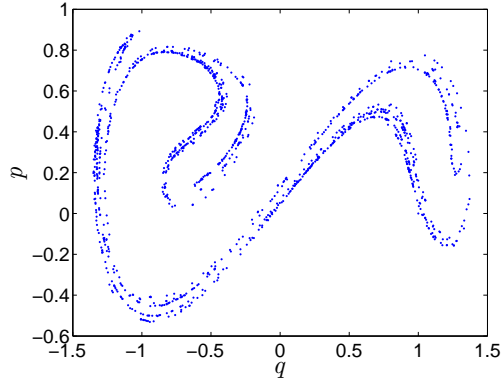}
  \caption{Poincar\'e section of the forced Duffing oscillator with damping at phase $\phi=0$. The oscillator displays chaotic dynamics and the attractor above displays the stretching and folding properties of chaos~\cite{Gucken:book}. \label{Fig:ForcedDuffing}}
\end{figure}

 Let us now assume that $\lambda$ is normally distributed. Let $\lambda = \lambda_{0} + \sigma\eta$, where $\lambda_{0} = -1.0$ is the mean of $\lambda$ and $\sigma = 0.1$ is its standard deviation. Is is easy to see that $\eta$ will now become a normally distributed random variable with zero mean and unit standard deviation. Since $\eta$ is normally distributed, we use Hermite polynomials in our expansion~\cite{Wiener}. In Eq.~\ref{eq:duffing2d} we use the expansion in Eq.~\ref{eq:expx}.  Truncating the expansion at $r=1$ gives the following set of differential equations:
\begin{equation}
\begin{pmatrix}
\dot Q_{0}\\
\dot P_{0}\\
\dot Q_{1}\\
\dot P_{1}
\end{pmatrix}=\begin{pmatrix}
P_{0}\\
- \delta P_{0} - \lambda_{0}Q_{0} - \sigma Q_{1} - (Q_{0}^3 + 3Q_{0}Q_{1}^2) + \gamma \cos\omega t\\
P_{1}\\
- \delta P_{1} - \lambda_{0}Q_{1} - \sigma Q_{0} - 3(Q_{1}^3 + Q_{0}^2Q_{1}).
\end{pmatrix},
\label{eq:a}
\end{equation}
Note that there is nothing special about order $r=1$, and the same procedure can be repeated for any $r$. The initial condition on the generalized coordinates $q$ and conjugate momenta $p$ gets incorporated into the initial conditions on the coefficients of expansion: $Q_{i}$ and $P_{i}$. The Poincar\'e section for Eqs.~\ref{eq:a} are shown in Fig.~\ref{Fig:Coeff_certain}. In this case, the stretching and folding structure of the Duffing oscillator is not as evident as in Fig.~\ref{Fig:ForcedDuffing}. However, the resulting dynamical system in Eqs.~\ref{eq:a} has a Lyapunov exponent of $\Xi \approx 0.73$, implying the persistence of sensitive dependence to initial conditions. Note that the route to chaos~\cite{Gucken:book} for the the original Duffing oscillator is well known. In~\cite{Cit:DuffingR1,Cit:DuffingR2}, the authors numerically demonstrate that the forced Duffing oscillator becomes chaotic due to a sequence of period doubling bifurcations. Due to the onset of chaos, the solution becomes increasingly difficult for polynomial chaos to capture. We point out that polynomial chaos is known to suffer from an inability to track output distributions for long term simulations~\cite{Cit:long-term}.

 The reason for the inability of polynomial chaos to track the output distribution lies in the increasingly oscillatory nature of the solution $q(t;\lambda)$ in terms of the uncertain parameter $\lambda$. In other words, any finite expansion in Eqn.~\ref{eq:expx} will fail at some $t$, since $q(t,\lambda)$ is too oscillatory in terms of $\lambda$. The greater the oscillatory nature of the output in terms of $\lambda$, the worse polynomial chaos performs. In~\cite{hybrid_uq}, the oscillatory nature of the output is again found to adversely impact the propagation of uncertainty through hybrid dynamical systems. However, we find that chaotic dynamics exacerbates this phenomenon. In particular, due to the coexistence of periodic orbits of different periods along with the chaotic attractor, the solution is found to rapidly become oscillatory with respect to $\lambda$ (depicted in Fig.~\ref{Fig:Duffing_osc}). Polynomial chaos is unable to track the first moment (mean) of $q(t;\lambda)$ beyond certain time (10~secs in Fig.~\ref{Fig:Duffing_tracking}, 25~secs in Fig.~\ref{Fig:Duffing_tracking_IC4}). In contrast to the forced Duffing oscillator, polynomial chaos is able to accurately track the mean of $q$ in the undamped and unforced Duffing oscillator with an order of expansion of $r=1$ (see Figs.~\ref{Fig:Duffing_simple_tracking} and~\ref{Fig:Duffing_simple_tracking_IC4}). Note that all parameters and initial conditions are held constant here (except for the removal of the forcing and damping terms). Hence, one needs to be careful when applying polynomial chaos to systems that are chaotic. We point to a caveat that if the initial condition is chosen close to the stable and unstable manifolds of the saddle equilibrium $(0,0)$, polynomial chaos performs poorly on the undamped, unforced oscillator case due to the discontinuity associated with the basin boundary~\cite{hybrid_uq}.

\begin{figure}
  \centering
  \subfigure[]{\includegraphics[scale=0.4]{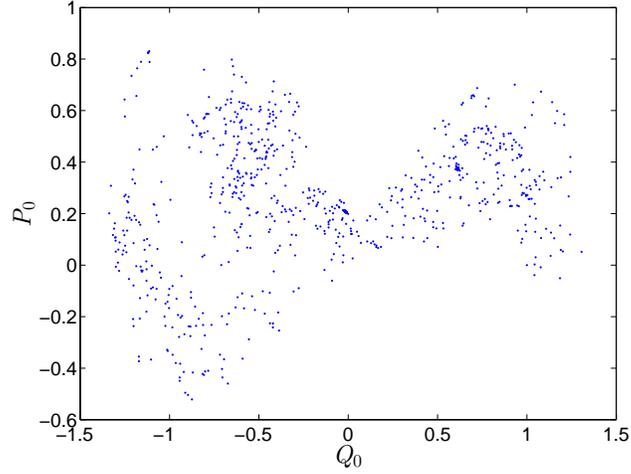}\label{Fig:Coeff0_certain}}
  \subfigure[]{\includegraphics[scale=0.4]{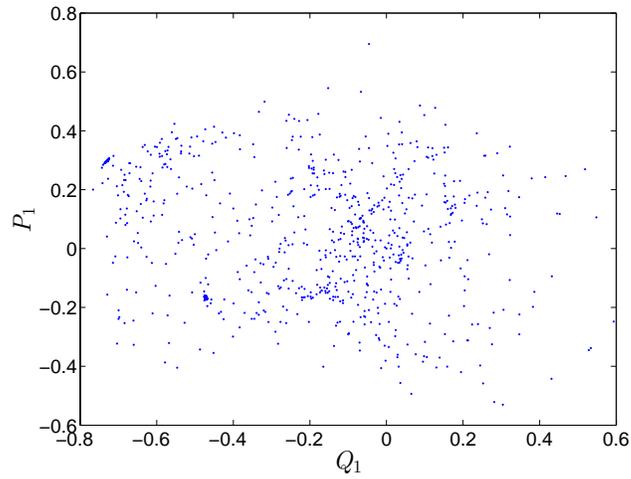}\label{Fig:Coeff1_certain}}
  \caption{Poincar\'e section at $\phi = 0$ of the dynamical system with certain initial conditions for the $0$-th (top) and first (bottom) order coefficients in the polynomial chaos expansion.\label{Fig:Coeff_certain}}
\end{figure}

\begin{figure}
  \centering
  \includegraphics[scale=0.4]{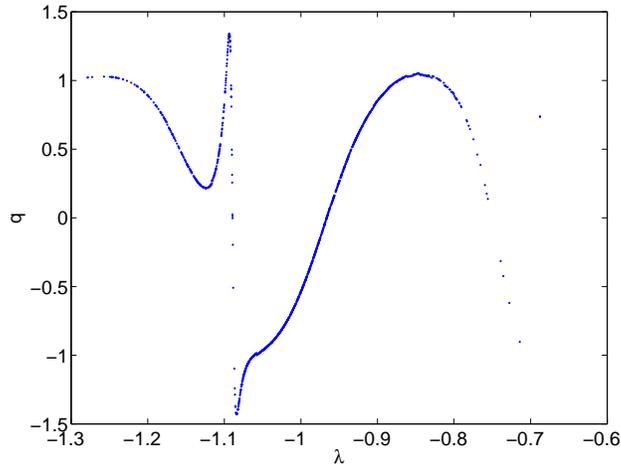}
  \caption{$q(t;\lambda)$ as a function of $\lambda$ for the Duffing oscillator at $t\approx 15$ sec. The solution is already too oscillatory in terms of $\lambda$ for an expansion to $r=1$.\label{Fig:Duffing_osc}}
\end{figure}

\begin{figure}
  \centering
  \includegraphics[scale=0.4]{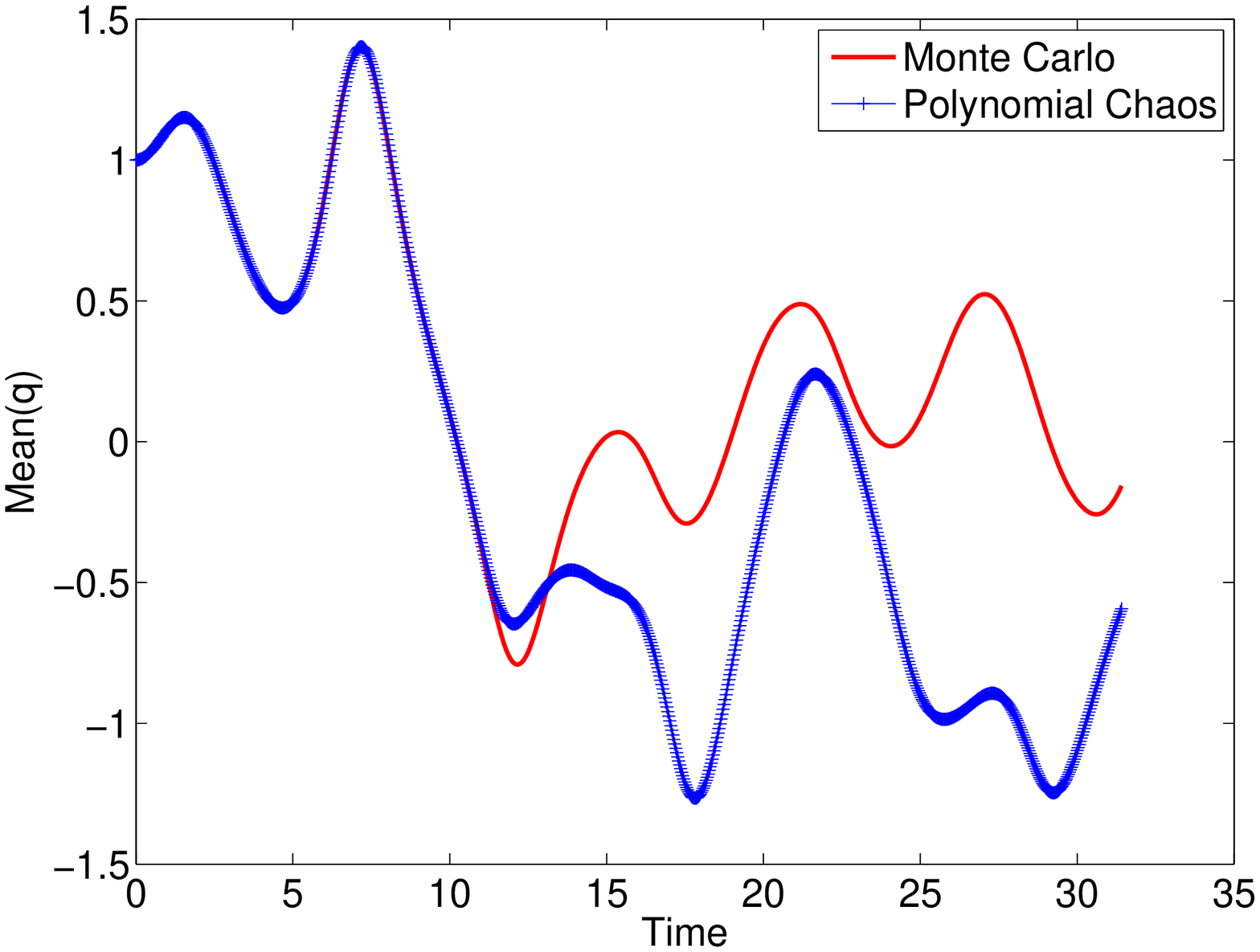}
  \caption{Comparison of Monte Carlo with polynomial chaos for the mean of $q$ as a function of time in the Duffing oscillator with initial condition $(q,p)=(1,0)$. After $t\approx 10$s, polynomial chaos (expansion to $r=1$) is unable to accurately track the first moment of the output distribution.\label{Fig:Duffing_tracking}}
\end{figure}

\begin{figure}
  \centering
  \includegraphics[scale=0.4]{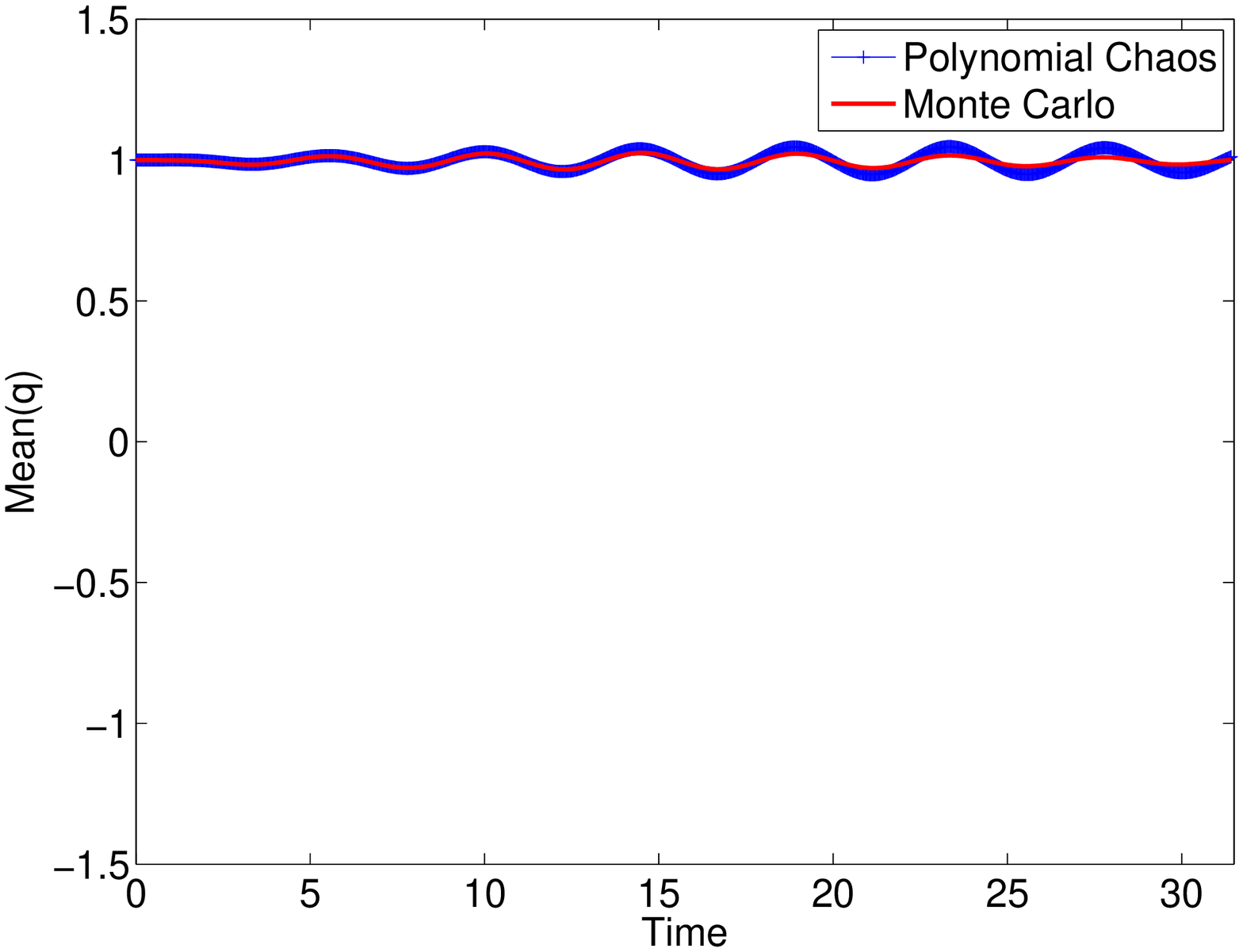}
  \caption{Comparison of Monte Carlo with polynomial chaos for the mean of $q$ as a function of time in the undamped, unforced Duffing oscillator with initial condition $(q,p)=(1,0)$. Polynomial chaos (expansion to $r=1$) is able to accurately track the first moment of the output distribution.\label{Fig:Duffing_simple_tracking}}
\end{figure}

\begin{figure}
  \centering
  \includegraphics[scale=0.4]{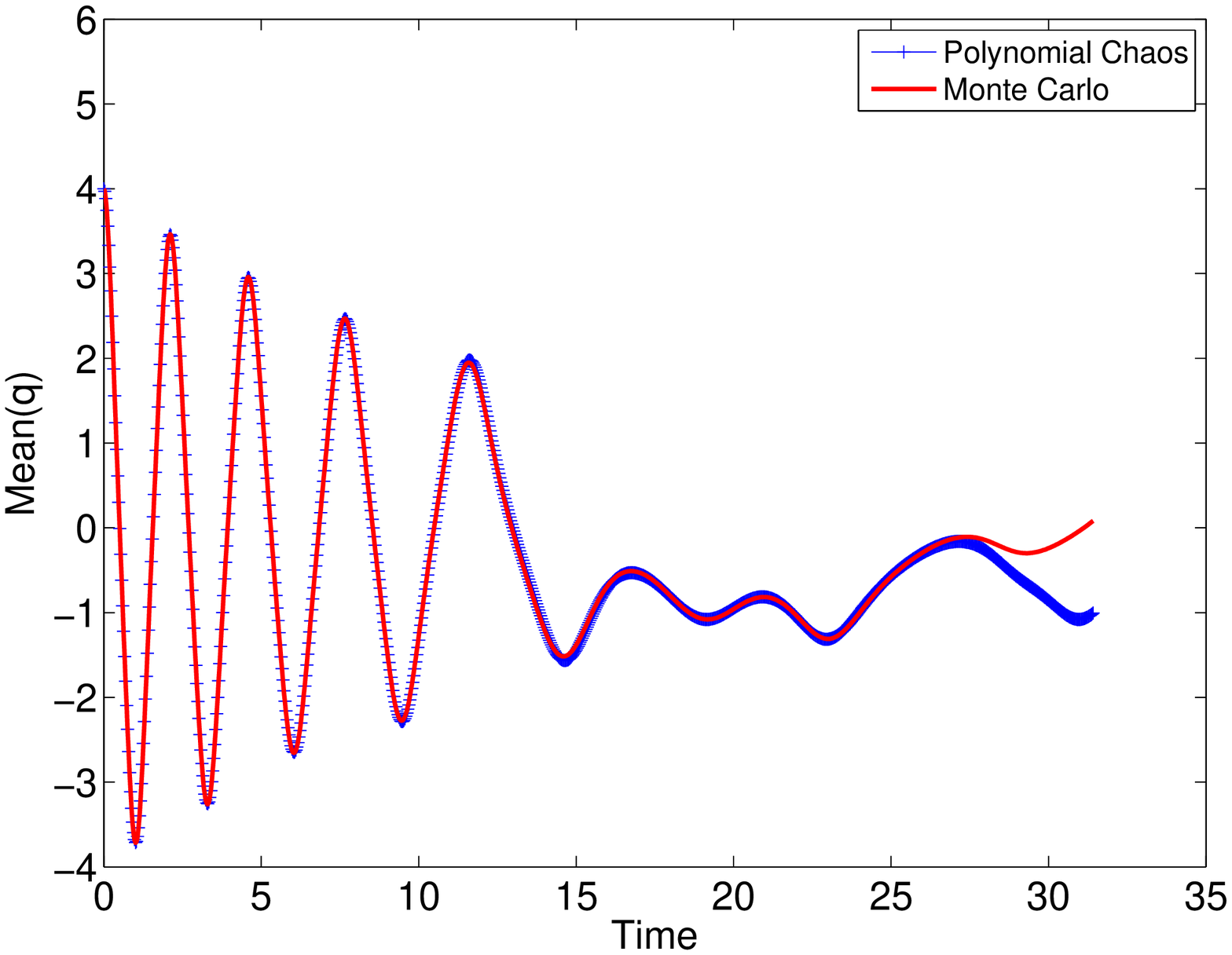}
  \caption{Comparison of Monte Carlo with polynomial chaos for the mean of $q$ as a function of time in the Duffing oscillator with initial condition $(q,p)=(4,0)$. After $\approx t=25$s, polynomial chaos (expansion to $r=1$) is unable to accurately track the first moment of the output distribution.\label{Fig:Duffing_tracking_IC4}}
\end{figure}

\begin{figure}
  \centering
  \includegraphics[scale=0.4]{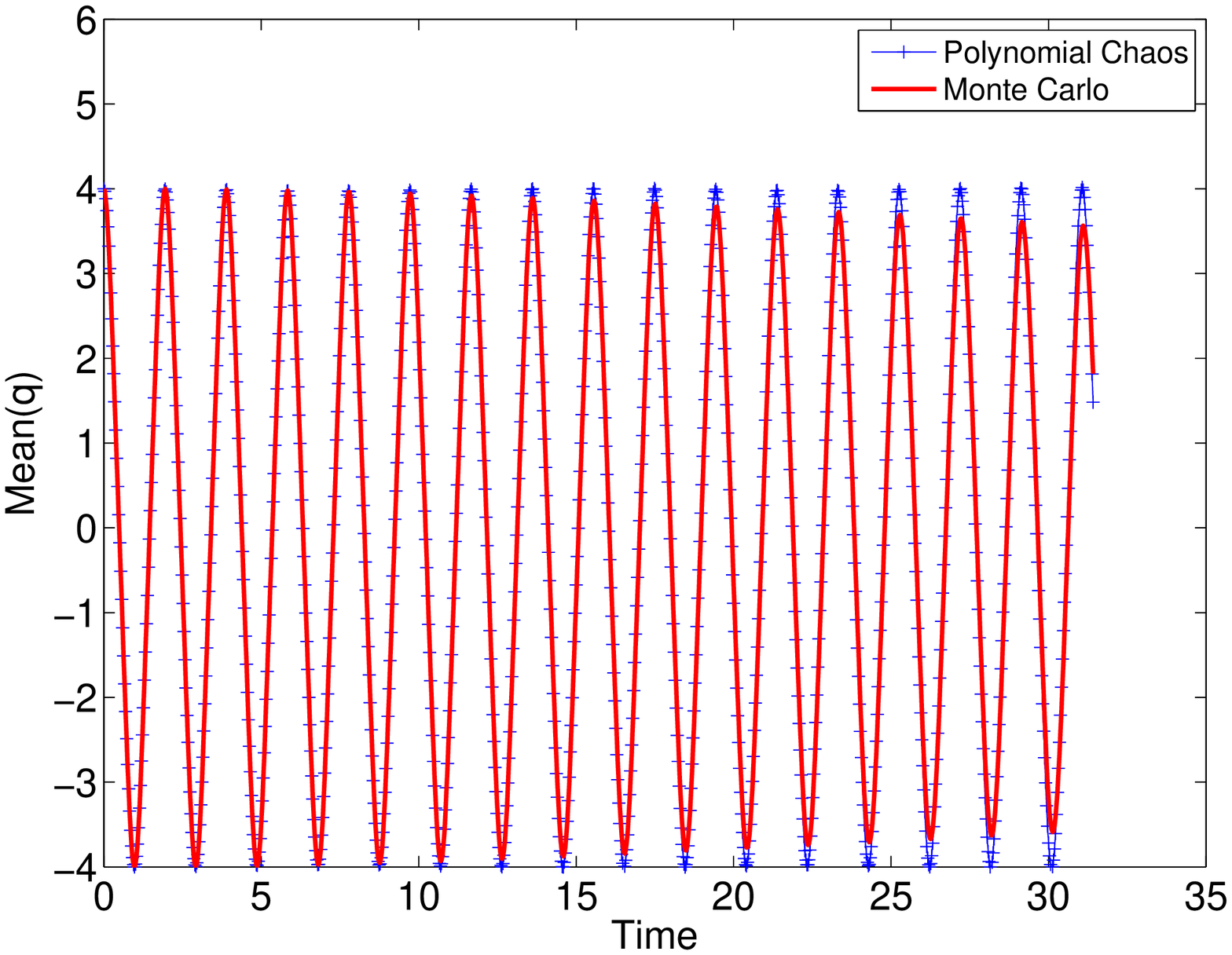}
  \caption{Comparison of Monte Carlo with polynomial chaos for the mean of $q$ as a function of time in the undamped, unforced Duffing oscillator with initial condition $(q,p)=(4,0)$. Polynomial chaos (expansion to $r=1$) is able to accurately track the first moment of the output distribution.\label{Fig:Duffing_simple_tracking_IC4}}
\end{figure}

When propagating uncertainty through chaotic systems with uncertain initial conditions, polynomial chaos again will need to be used carefully. Assume that $\lambda$ is not uncertain anymore, but instead the \emph{initial conditions} are normally distributed as $(q,p) = (1,0) + (\sigma\eta,0)$, where $\eta$ is a Gaussian variable with zero mean and unit variance.  The first order expansion yields the following system:
\begin{equation}
\begin{pmatrix}
\dot Q_{0}\\
\dot P_{0}\\
\dot Q_{1}\\
\dot P_{1}
\end{pmatrix}=\begin{pmatrix}
P_{0}\\
- \delta P_{0} - \lambda Q_{0} - (Q_{0}^3 + 3Q_{0}Q_{1}^2) + \gamma \cos \omega t\\
P_{1}\\
- \delta P_{1} - \lambda Q_{1} - 3(Q_{1}^3 + Q_{0}^2Q_{1}).
\end{pmatrix},
\label{eq:ic}
\end{equation}
Note that the uncertainty in initial conditions does not appear explicitly in these equations, but rather enters through the initial condition $Q_1(0) = \sigma$. For the purpose of our simulations we take $\sigma=0.1$. The resulting Poincar\'e sections are shown in Fig.~\ref{Fig:Coeff_uncertain}. The Lyapunov exponent is numerically found to be $\approx 0.85$, suggesting the persistence of chaos in the resulting polynomial chaos equations. This implies that any long term simulation that aims to track the output distribution will also suffer from problems of round-off in the initial conditions (given that the distribution on the initial condition will require computation of the initial conditions of the coefficients).

\begin{figure}
  \centering
  \subfigure[]{\includegraphics[scale=0.4]{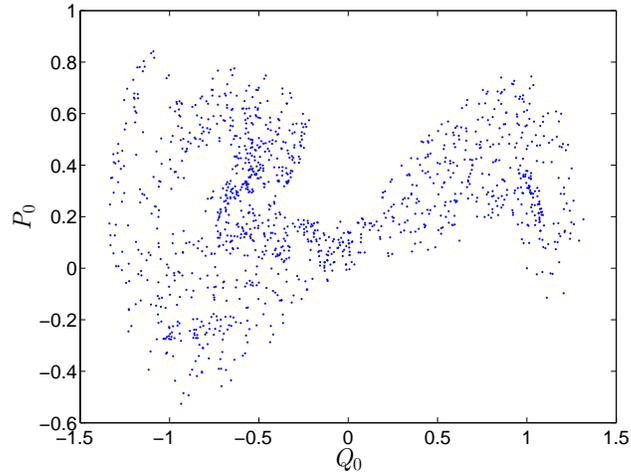}\label{Fig:Coeff0_uncertain}}
  \subfigure[]{\includegraphics[scale=0.4]{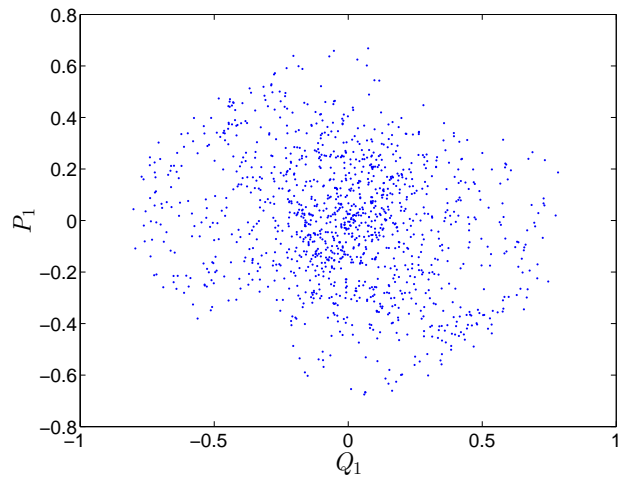}\label{Fig:Coeff1_uncertain}}
  \caption{a) Poincar\'e section at $\phi = 0$ of the dynamical system with uncertain initial conditions for the $0$-th order coefficients in the polynomial chaos expansion.  b) Poincar\'e section at $\phi = 0$ of the dynamical system with uncertain initial conditions for the $1$-st order coefficients in the polynomial chaos expansion.\label{Fig:Coeff_uncertain}}
 \end{figure}

\section{Conclusions}
Polynomial chaos is slowly becoming an established and popular approach for propagating uncertainty through smooth systems. Every year researchers use the approach to propagate uncertainty through a wide variety of engineering~\cite{Allen2009,Ghanem1998,Najm2009, Cit:Materials} and biological systems~\cite{Cit:polybio}. A systematic study on the properties and applicability of polynomial chaos to systems based on their structure and dynamics appears to be lacking.

In this work, we presented three main results. In the first part, we proved that when polynomial chaos is applied to Hamiltonian systems, the resulting equations are also Hamiltonian, even when the expansion is truncated. This is important, as it implies that structure in Hamiltonian systems is not only inherited by the new equations but also require the use of structure-preserving integrators~\cite{Cit:geom} to accurately propagate uncertainty. We also used the volume-preserving property of Hamiltonian systems to show that, on a particular example, a finite expansion must fail at long times, regardless of the order of the expansion. In the second part, we show that polynomial chaos may be applied to the averaged equations of a forced two-time system, allowing much faster uncertainty propagation than polynomial chaos on the original system. As the time scale separation increases, both the computational advantage as well as the quality of the approximation improves. In the third part, we demonstrate that polynomial chaos also inherits chaotic dynamics from underlying systems. The presence of chaos is shown to negatively influence the applicability of polynomial chaos. It reduces the length of time that polynomial chaos accurately tracks the output distributions and complicates computations when there is uncertainty in initial conditions.

\section*{Acknowledgments}

The authors would like to thank Amit Surana, as well as the anonymous reviewers for constructive comments.

\bibliographystyle{unsrt}
\bibliography{Hamiltonian}

\end{document}